\documentclass[12pt, draftclsnofoot, onecolumn]{IEEEtran}
\ifCLASSINFOpdf
\else
\fi
\usepackage{xcolor}
\usepackage{amsmath}
\usepackage{amsfonts}
\usepackage{amssymb}
\usepackage{amsthm}
\usepackage{graphicx}
\usepackage{lipsum}
\usepackage{cite}
\newenvironment{claim}[1]{\par\noindent\it{\textbf{Proposition~1:}}\space#1}{}
\newenvironment{lemma}[1]{\par\noindent\it{\textbf{Lemma~1:}}\space#1}{}

\newcommand{\G}{{\bf G}}
\newcommand{\T}{{\bf T}}
\newcommand{\R}{{\bf R}}

\begin{document}
%
\title{Power Adaptation for Vector Parameter Estimation according to Fisher Information based Optimality Criteria}

\author{Do\u{g}a~G\"{u}rg\"{u}no\u{g}lu,
        Berkan~Dulek,
        and Sinan~Gezici\thanks{D. G\"{u}rg\"{u}no\u{g}lu is with the Division of Decision and Control Systems, Electrical Engineering and Computer Science, KTH Royal Institute of Technology, 114 28 Stockholm, Sweden (e-mail: dogag@kth.se). B. Dulek is with the Department of Electrical and Electronics Engineering, Hacettepe University, Beytepe Campus, Ankara 06800, Turkey, (e-mail: berkan@ee.hacettepe.edu.tr). S. Gezici is with the Department of Electrical and Electronics Engineering, Bilkent University, Ankara 06800, Turkey (e-mail: gezici@ee.bilkent.edu.tr).}
}

\maketitle

\begin{abstract}
The optimal power adaptation problem is investigated for vector parameter estimation according to various Fisher information based optimality criteria. By considering an observation model that involves a linear transformation of the parameter vector and an additive noise component with an arbitrary probability distribution, six different optimal power allocation problems are formulated based on Fisher information based objective functions. Via optimization theoretic approaches, various closed-form solutions are derived for the proposed problems. Also, the results are extended to cases in which nuisance parameters exist in the system model or certain types of nonlinear transformations are applied on the parameter vector. Numerical examples are presented to investigate performance of the proposed power allocation strategies.
\end{abstract}


%
\IEEEpeerreviewmaketitle

\newpage

\section{Introduction}
%
%
%
%

In vector parameter estimation, the aim is to design an optimal estimator for a number of unknown parameters based on a set of observations.
The design of an optimal estimator commonly involves the calculation of posterior distributions or likelihood functions based on the statistical relation between the observation and the parameter vector.
If the prior distribution of parameters is known, the Bayesian approach can be adopted and estimators such as the minimum mean squared error (MMSE) estimator, the minimum mean absolute error (MMAE) estimator, or the maximum \emph{a posteriori} probability (MAP) estimator can be derived based on the posterior distribution, i.e., the probability distribution of the parameter vector given the observation \cite{kayestimation}.
On the other hand, in the absence of prior information, parameters can be modeled as a deterministic unknown vector and estimators such as
the minimum variance unbiased estimator (MVUE), the maximum likelihood (ML) estimator, or the best linear unbiased estimator (BLUE) can be employed for vector parameter estimation \cite{poor}.

Performance of the aforementioned estimators depends on system parameters such as noise variance and transformations acting on the parameter vector, and it is usually challenging to find exact and closed-form expressions for estimation errors of the corresponding estimators. Therefore, in order to assess estimation performance, various theoretical bounds such as the Cram\'{e}r-Rao lower bound (CRLB), Ziv-Zakai lower bound (ZZLB),
and Barankin-type bounds are used as gold standards \cite{Nehorai2013}.
Such bounds are mainly determined by the statistics of the observation, which depends on system parameters. This means that for a given system model, estimation performance can be improved only to a certain extent by using an optimal estimator.
In order to realize further improvements in estimation performance,
the effects of the system on the parameter vector should be adapted.
One common way of achieving such an improvement is to perform
power adaptation, i.e., transmitting different components of the parameter vector with different power levels \cite{resourceAllocDigit}. Since the total available power is usually limited \cite{distEstEnConstr}, the problem of power adaptation arises as a constrained optimization problem.
In this manuscript, the aim is to develop optimal power allocation strategies for vector parameter estimation in the absence of prior information by using Fisher information based optimality criteria 
\cite[Section IV.E.1]{poor}, \cite{wcevar_cite}, \cite[Section 9.2.1]{kalaba}, .

Power adaptation and in general resource allocation have been considered for various estimation problems in the literature. For example, in wireless sensor networks (WSNs), the problem of optimal resource allocation for vector parameter estimation with respect to various performance metrics is the main focus in many studies.
In \cite{distVecPowBWConstr}, the optimal transmit power allocation and quantization rate allocation schemes are investigated to minimize the average mean squared error (MSE). In \cite{distBlueFeedbackCSI}, the optimal power allocation strategy that minimizes the $\ell_2$-norm of the transmit power vector is derived under a maximum variance constraint for the best linear unbiased estimator. In addition, the optimal codebook is computed via the Lloyd algorithm when the channel state information (CSI) is limited, which is usually the case for large WSNs. In \cite{resourceAllocDigit}, estimation of an unknown Gaussian random vector with known mean vector and covariance matrix is considered in a WSN setting, where the fusion center uses the linear MMSE (LMMSE) estimator to estimate the parameter vector based on sensor observations, which are fading channel impaired and noise corrupted versions of the transmitted parameter vector. An upper bound on the MSE is minimized by first computing the optimal bit allocation to minimize the MSE distortion. Then, the optimal power allocation strategy is computed to minimize the channel errors. In \cite{shirazi}, optimal power allocation for vector parameter estimation is investigated with the aim of maximizing the average Bayesian Fisher information between the random parameter vector and the observation vector. In \cite{channelUncertainty}, optimal power allocation schemes for LMMSE estimation are derived by taking channel estimation errors into account. In \cite{tcom_11,tcom_12,tcom_13,tcom_14,tcom_15,tcom_16,tcom_17,tcom_18,tcom_19,tcom_20}, the optimal power allocation problem is considered for position estimation in wireless localization and radar systems. In \cite{tcom_13}, the transmit power allocation problem is formulated as a semidefinite program by using the squared position error bound as the objective function. In \cite{tcom_17}, the total transmit power is minimized by imposing a constraint on the CRLB for target localization in a distributed multiple-radar system. In addition, the dual problem of CRLB minimization for a predefined total power budget is considered.

It is noted that theoretical lower bounds for estimation error are commonly used in the literature to define optimality criteria for developing power adaptation strategies in estimation problems \cite{shirazi}, \cite{tcom_11,tcom_12,tcom_13,tcom_14,tcom_15,tcom_16,tcom_17,tcom_18,tcom_19,tcom_20,taes1,GeziciJammingTCOM,taes4,taes6}. In the absence of prior information, lower bounds generated from the Fisher information matrix (FIM) are usually adopted due to their practicality. As the most widely used bound, the CRLB is obtained as the inverse of the FIM and specifies a lower limit on the covariance matrix of any unbiased estimator with respect to the positive semidefinite cone. Various scalarizations of the FIM are employed in the literature \cite{tcom_17,GeziciJammingTCOM,taes6}.
In particular, the log-determinant of the FIM, the maximum (minimum) eigenvalue of the CRLB (FIM), the maximum  diagonal  entry  of  the  CRLB, the trace of the FIM, and  the minimum  diagonal  entry  of  the  FIM are utilized for quantifying estimation performance from various perspectives such as estimation robustness and probabilistic confinement of estimator error \cite[Section 9.2.1]{kalaba}, \cite{Emery_1998,boydbook,dogancay1,dogancay2,berkan,tzoreff}. In this manuscript, the power adaptation problem for vector parameter estimation is considered according to such Fisher information based optimality criteria and the corresponding optimal strategies are characterized.


Although there exist a multitude of studies on power allocation for various estimation problems in the literature, a general investigation of the optimal power allocation problem for vector parameter estimation according to various Fisher information based criteria is not available to the best of authors' knowledge. In particular, we consider a generic additive noise model, where the observation vector is a linear function of the parameter vector corrupted by additive noise with an arbitrary probability distribution.
Based on this model, we first present the FIM in terms of the system parameters, including the power allocation parameters. Then, we formulate optimal power allocation problems according to six different estimation performance criteria based on the FIM, and derive various closed-form solutions.
We also extend our results to cases in which nuisance parameters exist in the problem or certain types of nonlinear transformations are applied on the parameter vector. The main contributions and novelty of this manuscript can be summarized as follows:
\begin{itemize}
    \item According to various Fisher information based optimality criteria, we propose optimal power allocation problems for vector parameter estimation by considering a system model, where the parameter vector is processed by a linear transformation and corrupted by additive noise with a generic probability distribution.
    \item Based on optimization theoretic approaches, we provide various closed-form solutions for the proposed power allocation problems.
    \item We show that the proposed optimal power allocation strategies are also valid {for nonlinear system models under certain conditions and} in the presence of nuisance parameters.
\end{itemize}
In addition, we provide numerical examples to illustrate the performance of the proposed strategies and compare them with the equal power allocation strategy. It should be noted that providing closed-form solutions for optimal power allocation is important for real-time applications due to delay and computational complexity requirements.


The rest of the manuscript is organized as follows: The system model is presented in Section~\ref{sec:sysmodel} and optimal power allocation strategies are derived in Section~\ref{sec:approaches}. In Section~\ref{sec:Extend}, extensions
to nonlinear models and presence of nuisance parameters are considered. 
Numerical results are provided in Section~\ref{sec:results} followed by the concluding remarks in Section~\ref{sec:conclusion}.

\section{System Model}\label{sec:sysmodel}

Consider the following linear\footnote{Extensions to nonlinear models are presented in Section~\ref{sec:Extend}.} model relating a vector of unknown deterministic parameters $\boldsymbol{\theta}=[\theta_1,\ldots,\theta_k]^T\in\mathbb{R}^{k}$ with their measurements ${\bf{X}}\in\mathbb{R}^n$:
\begin{equation}\label{eq:sysmodel}
    {\bf{X}}={\bf{F}}^T{\bf{P}}{\boldsymbol{\theta}}+{\bf{N}}
\end{equation}
In \eqref{eq:sysmodel}, $\bf{F}$ is a $k\times n$ real matrix with full row rank ($k\leq n$) that is assumed to be known, ${\bf{N}}\in\mathbb{R}^n$ is the additive noise vector with a joint probability density function $f_{\bf{N}}(\cdot)$, which is independent of $\boldsymbol{\theta}$, and $\bf{P}$ is a $k\times k$ diagonal power allocation matrix (to be optimized) expressed as
\begin{equation}\label{eq:Pmatrix}
    \bf{P} = \begin{bmatrix}
                \sqrt{p_1} & & \bf{0} \\
                 & \ddots & \\
                \bf{0} & & \sqrt{p_k}
             \end{bmatrix}
\end{equation}
subject to the total power constraint
\begin{equation}
    \sum_{i=1}^{k}{p_i}\leq P_{\Sigma}
\end{equation}
where $p_i$ denotes the power allocated to the parameter $\theta_i$ and $P_{\Sigma}$ denotes the (available) total power. For the linear model given in \eqref{eq:sysmodel}, the FIM of the measurement vector $\bf{X}$ with respect to the parameter vector $\boldsymbol{\theta}$ is obtained as \cite[Lemma 5]{zamir}.
\begin{equation}\label{eq:fimx}
    {\bf{I}}({\bf{X}};\boldsymbol{\theta})={\bf{P}}{\bf{F}}{\bf{I}}({\bf{N}}){\bf{F}}^T{\bf{P}},
\end{equation}
where ${\bf{P}}={\bf{P}}^T$ is employed, and ${\bf{I}}({\bf{N}})$ is a special form of the FIM, namely the FIM of the random vector $\bf{N}$ with respect to a translation parameter $\boldsymbol{\phi}$ \cite[Equation 8]{zamir}, defined as
\begin{equation}\label{eq:In}
    {\bf{I}}({\bf{N}})={\bf{I}}(\boldsymbol{\phi}+{\bf{N}};\boldsymbol{\phi})
    =\int{\frac{1}{f_{\bf{N}}({\bf{n}})}\left(\frac{\partial f_{\bf{N}}({\bf{n}})}{\partial {\bf{n}}}\right)\left(\frac{\partial f_{\bf{N}}({\bf{n}})}{\partial {\bf{n}}}\right)^Td{\bf{n}}}
\end{equation}
It is noted that the FIM under translation is a function of only the probability density function (pdf) of the random vector $\bf{N}$, and consequently, ${\bf{I}}({\bf{X}};\boldsymbol{\theta})$ in \eqref{eq:fimx} does not depend on the parameter vector $\boldsymbol{\theta}$. {It is assumed that the noise pdf $f_{\bf{N}}(\cdot)$ satisfies certain regularity conditions so that the FIM in \eqref{eq:In} exists \cite{schervish}.}

In the following, we provide closed-form solutions for optimal power allocation problems by considering various estimation accuracy criteria based on the FIM in \eqref{eq:fimx}.


\section{Optimal Power Allocation for Vector Parameter Estimation}\label{sec:approaches}

\subsection{Average Mean Squared Error Criterion}\label{sec:AvgMSE}

The inverse of the FIM, known as the Cramer-Rao lower bound (CRLB)  provides a lower bound on the MSE of any unbiased estimator $\hat{\boldsymbol{\theta}}({\bf{X}})$ via the following matrix inequality \cite{poor}:
\begin{equation}
    {\rm{Cov}}(\hat{\boldsymbol{\theta}}({\bf{X}})) \geq {\bf{I}}^{-1}({\bf{X}};{\boldsymbol{\theta}})
\end{equation}
where ${\rm{Cov}}(\hat{\boldsymbol{\theta}}({\bf{X}}))= {\rm{E}}[(\hat{{\boldsymbol{\theta}}}({\bf{X}})-{\boldsymbol{\theta}})(\hat{{\boldsymbol{\theta}}}({\bf{X}})-{\boldsymbol{\theta}})^T]$ due to the unbiasedness and the expectation is taken with respect to the pdf of $\bf{X}$ given $\boldsymbol{\theta}$.
Consequently, the lower bound on the average MSE of the vector parameter can be stated as
\begin{equation}\label{eq:AoptCri}
    {\rm{E}}\big[\rVert\hat{{\boldsymbol{\theta}}}({\bf{X}})-{\boldsymbol{\theta}}\lVert^2\big] \geq {\rm{tr}}\{{\bf{I}}^{-1}({\bf{X}};{\boldsymbol{\theta}})\}
\end{equation}
Consideration of the lower bound in \eqref{eq:AoptCri} as a performance metric in optimal design problems is referred as the A-optimality criterion in the literature \cite{kalaba,Emery_1998,taes1}.

The optimal power allocation problem that minimizes the lower bound on the average MSE subject to a sum-power constraint can be formulated as
\begin{equation}\label{eq:avgmseopt}
    \begin{aligned}
        \min_{\{p_i\}_{i=1}^{k}}\quad&{\rm{tr}}\{{\bf{I}}^{-1}({\bf{X}};\boldsymbol{\theta})\} \\
        \textrm{s.t.}\quad &\sum_{i=1}^{k}{p_i} \leq P_{\Sigma}\\
        &p_i \geq 0,\quad i=1,2,\ldots,k
    \end{aligned}
\end{equation}
For the convenience of notation, two system dependent matrices can be defined as
\begin{equation}\label{eq:jdef}
\begin{aligned}
        {\bf{J}} &\triangleq {\bf{F}}{\bf{I}}({\bf{N}}){\bf{F}}^T\\
        {\bf{A}}&\triangleq{\bf{J}}^{-1}
\end{aligned}
\end{equation}
From \eqref{eq:fimx} and \eqref{eq:jdef}, the FIM with respect to the parameter vector $\boldsymbol{\theta}$ and the corresponding CRLB are expressed respectively as ${\bf{I}}({\bf{X}};{\boldsymbol{\theta}})={\bf{P}}{\bf{J}}{\bf{P}}$ and ${\bf{I}}^{-1}({\bf{X}};{\boldsymbol{\theta}})={\bf{P}}^{-1}{\bf{A}}{\bf{P}}^{-1}$. Then, the objective function in \eqref{eq:avgmseopt} can be written in terms of the power allocation coefficients and the diagonal entries of $\bf{A}$ as
\begin{equation}\label{eq:ObjFn1}
    \begin{aligned}
        {\rm{tr}}\{{\bf{I}}^{-1}({\bf{X}};\boldsymbol{\theta})\} &= \rm{tr}\{{\bf{P}}^{-1}{\bf{A}}{\bf{P}}^{-1}\}\\
        &=\rm{tr}\{({\bf{P}}^{-1})^2\bf{A}\}\\
        &=\sum_{i=1}^{k}{\frac{a_{ii}}{p_{i}}}\\
    \end{aligned}
\end{equation}
where $a_{ii}$ denotes the $i$th diagonal entry of $\bf{A}$. It is noted that the FIM is assumed to be positive-definite for the existence of the CRLB. Therefore, ${\bf{A}}$ and ${\bf{J}}$ in \eqref{eq:jdef} have positive diagonal entries.

As the objective function is convex (see \eqref{eq:ObjFn1}) and the constraints are linear, the problem in \eqref{eq:avgmseopt} is a convex optimization problem. In addition, Slater's condition holds \cite{boydbook}. Therefore, Karush-Kuhn-Tucker (KKT) conditions are necessary and sufficient for optimality. From \eqref{eq:ObjFn1}, the Lagrangian for \eqref{eq:avgmseopt} is expressed as
\begin{align}
\label{eq:Lagr}
        \mathcal{L}\big(\{p_i\}_{i=1}^{k},\{\upsilon_i\}_{i=1}^{k+1}\big)=\sum_{i=1}^{k}{\frac{a_{ii}}{p_i}}+
        \upsilon_1\left(\sum_{i=1}^{k}{p_i}-P_\Sigma\right)
        -\sum_{i=1}^{k}\upsilon_{i+1}p_{i}
\end{align}
where $\upsilon_1,\ldots,\upsilon_{k+1}$ are the dual variables.
Then, KKT conditions for optimality are obtained as follows \cite{boydbook}:
\begin{itemize}
    \item \textit{Primal Feasibility:} The optimal power allocation strategy $\{p_i^*\}_{i=1}^{k}$ must satisfy the constraints $\sum_{i=1}^{k}{p_i^*}\leq P_\Sigma$ and $p_i^*\geq 0,\:\forall i\in\{1,\ldots,k\}$.
    \item \textit{Dual Feasibility:} The dual variables must be non-negative, i.e., $\upsilon_i^*\geq 0$ for $i=1,\ldots,k+1$.
    \item \textit{Stationarity:} The derivatives of the Lagrangian in \eqref{eq:Lagr} with respect to $p_i$ must be equal to zero at $p_i=p_i^*$ for $i=1,\ldots,k$. That is,
\begin{equation}\label{eq:avgmseinter}
        \frac{\partial \mathcal{L}}{\partial p_i}\bigg|_{p_i=p_i^*} = -\frac{a_{ii}}{(p_i^*)^2}+{\upsilon}_1^*-{\upsilon}_{i+1}^*=0
\end{equation}
for $i=1,\ldots,k$.
    \item \textit{Complementary Slackness:} At the optimal solution, the following conditions hold:
    \begin{align}\label{eq:CompSlack1}
        {\upsilon}_1^*\left(\sum_{i=1}^{k}{p_i^*}-P_\Sigma\right)&=0\\\label{eq:CompSlack2}
        {\upsilon}_{i+1}^*p_{i}^*&=0,~i=1,\ldots,k
    \end{align}
\end{itemize}
For the condition in \eqref{eq:CompSlack1}, the case of ${\upsilon}_1^*=0$ is not possible since the derivative in \eqref{eq:avgmseinter} could be set to zero only for $p_i^*\rightarrow\infty$ in that case (for some positive $a_{ii}$), which would violate the primal feasibility condition. Therefore, \eqref{eq:CompSlack1} implies that ${\upsilon}_1^*>0$ and
\begin{equation}\label{eq:fullpower}
\sum_{i=1}^{k}{p_i^*}=P_\Sigma\,
\end{equation}
That is, full-power utilization is required for optimality\footnote{This fact can also be seen by noting that the objective function in \eqref{eq:ObjFn1} is a decreasing function of $p_i$'s.}. Then, two cases are investigated depending on the values of $a_{ii}$'s. Let $\mathcal{A}_z$ and $\mathcal{A}_p$ denote the sets of indices $i$ for which $a_{ii}$'s are zero and positive, respectively. That is, $\mathcal{A}_z=\{i\in\{1,\ldots,k\}\,|\,a_{ii}=0\}$ and $\mathcal{A}_p=\{i\in\{1,\ldots,k\}\,|\,a_{ii}>0\}$.

\underline{Case~1:} Consider an index $i$ such that $i\in\mathcal{A}_z$. Suppose that $p_i^*>0$. Then, \eqref{eq:CompSlack2} implies that $\upsilon_{i+1}^*=0$ and the expression in \eqref{eq:avgmseinter} becomes equal to $\upsilon_1^*$. However, $\upsilon_1^*>0$ as discussed before, which leads to a contradiction (i.e., the stationary condition could not be satisfied). Hence, it is concluded that
\begin{equation}\label{eq:avgmseinter0}
p_i^*=0\,~\textrm{for}\,~i\in\mathcal{A}_z.
\end{equation}

\underline{Case~2:} Consider an index $i$ such that $i\in\mathcal{A}_p$. In that case, it can be concluded from \eqref{eq:avgmseinter}--\eqref{eq:CompSlack2} that $p_i^*>0$ and $\upsilon_{i+1}^*=0$ for $i\in\mathcal{A}_p$. Then, \eqref{eq:avgmseinter} leads to
\begin{equation}\label{eq:avgmseinter2}
p_i^*=\sqrt{\frac{a_{ii}}{{\upsilon}_1^*}}\,~\textrm{for}\,~i\in\mathcal{A}_p
\end{equation}

From \eqref{eq:fullpower}, a relation for $\upsilon_1^*$ can be obtained as
\begin{equation}\label{eq:lambda1casea}
        \sum_{j=1}^{k}{p_j^*}=\sum_{j\in{\mathcal{A}}_p}{p_j^*}=\sum_{j\in{\mathcal{A}}_p}{\sqrt{\frac{a_{jj}}{{\upsilon}_1^*}}}=P_\Sigma,
\end{equation}
which yields
\begin{equation}\label{eq:lambda1casea2}
        \frac{1}{\sqrt{{\upsilon}_1^*}}=\frac{P_\Sigma}{\sum_{j\in{\mathcal{A}}_p}{\sqrt{a_{jj}}}}
        =\frac{P_\Sigma}{\sum_{j=1}^{k}{\sqrt{a_{jj}}}}
\end{equation}

Based on \eqref{eq:avgmseinter0}, \eqref{eq:avgmseinter2} and \eqref{eq:lambda1casea2}, the optimal power allocation strategy to minimize the average MSE in \eqref{eq:avgmseopt} is specified as follows:
\begin{equation}\label{eq:avgmse}
    p_i^*=\frac{P_\Sigma\sqrt{a_{ii}}}{\sum_{j=1}^{k}{\sqrt{a_{jj}}}},~~i=1,\ldots,k
\end{equation}
Hence, a closed-form solution to the problem in \eqref{eq:avgmseopt} is obtained.



\subsection{Shannon Information Criterion}

An alternative criterion for estimation accuracy is to maximize the log-determinant of the FIM, i.e.,
\begin{equation}\label{eq:logdet}
    {\bf{I}}_S({\bf{X}};\boldsymbol{\theta})=\log\det {\bf{I}}({\bf{X}};\boldsymbol{\theta})
\end{equation}
which is associated with the volume of the confidence ellipsoid containing the estimation error \cite[Section 7.5.2]{boydbook}. This criterion is known as the Shannon information criterion, and also the D-optimal design in the literature \cite{Emery_1998,dogancay1,dogancay2,taes1}.


The optimal power allocation problem with the objective of maximizing the Shannon information under the sum-power constraint can be expressed as
\begin{equation}\label{eq:shannonInfoOpt}
    \begin{aligned}
       \max_{\{p_i\}_{i=1}^{k}}\quad&\log\det {\bf{I}}({\bf{X}};\boldsymbol{\theta})\\
       \rm{s.t.}\quad&\sum_{i=1}^{k}{p_i} \leq P_\Sigma\\
        &p_i \geq 0,\quad i=1,\ldots,k
    \end{aligned}
\end{equation}
The problem in \eqref{eq:shannonInfoOpt} involves the maximization of a concave function, and the feasible region has an interior point; hence, Slater's condition is satisfied. Consequently, KKT conditions are necessary and sufficient for optimality. In order to find the optimal solution, the Shannon information can be expressed in terms of the known matrices. From \eqref{eq:fimx} and  \eqref{eq:jdef}, the Shannon information can be written as
\begin{equation}\label{eq:shannonInfoMatrix}
\begin{aligned}
    \log\det{{\bf{I}}({\bf{X}};\boldsymbol{\theta})} &= \log\det{\bf{P}\bf{J}\bf{P}}\\
    &= 2\log\det{\bf{P}}+\log\det{\bf{J}}\\
    &= \sum_{i=1}^{k}{\log{p_i}}+\log\det{\bf{J}}
\end{aligned}
\end{equation}
As seen in \eqref{eq:shannonInfoMatrix}, the Shannon information separates into a power allocation dependent component and a system dependent component, the latter being constant for a fixed $\bf{F}$ and ${\bf{I}}({\bf{N}})$. Therefore, it suffices to consider $\sum_{i=1}^{k}{\log{p_i}}$ in order to maximize the Shannon information. Therefore, the problem in \eqref{eq:shannonInfoOpt} reduces to
\begin{equation}\label{shannoninfosol}
    \begin{aligned}
        \max_{\{p_i\}_{i=1}^{k}}\quad&\sum_{i=1}^{k}{\log{p_i}}\\
        \rm{s.t.}\quad&\sum_{i=1}^{k}{p_i} \leq P_\Sigma\\
        &p_i \geq 0,\quad i=1,\ldots,k
    \end{aligned}
\end{equation}
which is a convex optimization problem. This problem is equivalent to maximizing the product of nonnegative numbers whose sum is constant. Hence, its solution can be obtained as
\begin{equation}\label{shannonOptimal}
    p_i^*=\frac{P_\Sigma}{k},\quad i=1,\hdots,k
\end{equation}
That is, the optimal power allocation strategy according to the Shannon information criterion is to allocate equal power to all the parameters at the sum-power limit. Corresponding to the optimal strategy, the maximum Shannon information is achieved as
\begin{equation}
    {\bf{I}}_S^*({\bf{X}};\boldsymbol{\theta})=k\log{\left(\frac{P_\Sigma}{k}\right)}+\log\det{\bf{J}}\,.
\end{equation}

\subsection{Worst-Case Error Variance Criterion}\label{sec:WCEVC}

The worst-case error variance criterion is a measure of robustness rather than average estimation accuracy and is associated with the maximum eigenvalue of the CRLB \cite{boydbook}, \cite{berkan}. In order to reduce the worst-case error variance, the maximum eigenvalue of the CRLB can be minimized. {Optimality according to this criterion is also known as E-optimality, where the minimum diameter of the FIM is maximized \cite{Emery_1998,tzoreff,taes1}. When variances vary significantly, the confidence ellipsoid can have very different diameters along different dimensions; hence, the log-volume minimization approach in the D-optimal design can be misleading \cite{tzoreff,Emery_1998}.}

The optimal power allocation strategy that minimizes the maximum eigenvalue of the CRLB corresponds to maximizing the minimum eigenvalue of the FIM. Hence, the following problem can be considered:
\begin{equation}\label{eq:wcevaropt}
    \begin{aligned}
        \max_{\{p_i\}_{i=1}^{k}}\quad&\lambda_{\min}\{{\bf{I}}({\bf{X}};\boldsymbol{\theta})\}\\
        \textrm{s.t.}\quad&\sum_{i=1}^{k}{p_i} \leq P_\Sigma\\
        &p_i \geq 0,\quad i=1,\ldots,k
    \end{aligned}
\end{equation}
From \eqref{eq:fimx} and \eqref{eq:jdef}, ${\bf{I}}({\bf{X}};\boldsymbol{\theta})$ in \eqref{eq:wcevaropt} can be expressed as ${\bf{I}}({\bf{X}};\boldsymbol{\theta})={\bf{P}}{\bf{J}}{\bf{P}}$, where ${\bf{J}}$ is positive semi-definite and ${\bf{P}}$ is diagonal (see \eqref{eq:Pmatrix}). It can be shown that the eigenvalues of ${\bf{P}}{\bf{J}}{\bf{P}}$ are the same as those of ${\bf{P}}^2{\bf{J}}$ based on their characteristic equations. However in general, there is not a closed-form relationship between the eigenvalues of ${\bf{P}}^2$ and ${\bf{J}}$ and the eigenvalues of their product.
Therefore, it is challenging to obtain a closed-form solution to \eqref{eq:wcevaropt}. One way to solve \eqref{eq:wcevaropt} is to apply global optimization tools such as particle swarm optimization (PSO) or the multistart algorithm \cite{PSO1}. This approach is adopted in Section~\ref{sec:results} to obtain the solution of \eqref{eq:wcevaropt}.

To perform further investigations on the problem in \eqref{eq:wcevaropt},
we can derive a bound on the objective function in \eqref{eq:wcevaropt}. To that aim, the following lemma can be utilized to provide bounds for the eigenvalues of the FIM, ${\bf{I}}({\bf{X}};\boldsymbol{\theta})$.

\vspace{0.1cm}

\begin{lemma}
Let $\G,\T \in \mathbb{R}^{n\times n}$, $\G^T=\G$, $\T^T=\T$, $rank(\G)=rank(\T)=n$ and $\R=\G\T$. Let $\{\nu_i\}_{i=1}^{n}$, $\{\mu_i\}_{i=1}^{n}$ and $\{\gamma_i\}_{i=1}^{n}$ denote, respectively, the eigenvalues of $\G$, $\T$ and $\R$ increasing in the absolute value. Then,
\begin{equation}\label{eq:EigenIneq}
    |\nu_1||\mu_1| \leq |\gamma_1| \leq \cdots \leq |\gamma_n| \leq |\nu_n||\mu_n|
\end{equation}
\end{lemma}

\vspace{-0.7cm}

\begin{proof}
For the Euclidean matrix norm, it is known that $\lVert \G\T \lVert \leq \lVert \G\lVert\lVert \T \lVert$. In addition, the norm of a symmetric matrix is equal to its spectral radius, i.e., $\lVert \G\lVert = \max_{j}|\nu_j|$ and $\lVert \T\lVert = \max_{j}|\mu_j|$. Therefore, $\lVert \R\lVert = \lVert \G\T\lVert \leq \lVert \G\lVert\lVert \T\lVert = |\nu_n||\mu_n|$. Since $\R$ is the product of two symmetric matrices, it is also symmetric. Therefore, $\lVert \R\lVert = |\gamma_n|$. Hence, the upper bound on the absolute value of the maximum eigenvalue of $\R$ is obtained as $|\gamma_n| \leq |\nu_n||\mu_n|$. Since  $rank(\G)=rank(\T)=n$, the lower bound can be derived via inversion; that is, $\R^{-1}=\T^{-1}\G^{-1}$. Therefore, $\lVert \R^{-1}\lVert=\lVert \T^{-1}\G^{-1}\lVert \leq \frac{1}{|\nu_1||\mu_1|}$. Through the same reasoning, $\lVert \R^{-1}\lVert = \frac{1}{|\gamma_1|}$. Hence, the relation of $\frac{1}{|\gamma_1|} \leq \frac{1}{|\nu_1||\mu_1|}$ is obtained, which yields the lower bound in \eqref{eq:EigenIneq}.
\end{proof}

Lemma~1 can be used to derive a lower bound on the objective function in \eqref{eq:wcevaropt} as follows:
\begin{equation}\label{eq:boundObj}
    \lambda_{\min}\{{\bf{I}}({\bf{X}};\boldsymbol{\theta})\}=\lambda_{\min}\{{\bf{P}}^2{\bf{J}}\}
    \geq \lambda_{\min}\{{\bf{J}}\}\min_{i\in\{1,\ldots,k\}}{p_i}
\end{equation}
where $\lambda_{\min}\{{\bf{J}}\}$ denotes the minimum eigenvalue of ${\bf{J}}$. In \eqref{eq:boundObj}, the absolute value operators in \eqref{eq:EigenIneq} are not used since all the eigenvalues are non-negative and the eigenvalues of ${\bf{P}}^2$ are taken as $\{p_1,\ldots,p_k\}$ based on \eqref{eq:Pmatrix}.

Instead of maximizing the minimum eigenvalue of ${\bf{I}}({\bf{X}};\boldsymbol{\theta})$ in \eqref{eq:wcevaropt}, consider the maximization of the lower bound on it.
As noted from \eqref{eq:boundObj}, the lower bound on the minimum eigenvalue of ${\bf{I}}({\bf{X}};\boldsymbol{\theta})$ depends on the minimum power allocated to an individual parameter. Therefore, instead of \eqref{eq:wcevaropt}, we get the following convex optimization problem:
\begin{equation}\label{eq:wcevarsol}
\begin{aligned}
    \max_{\{p_i\}_{i=1}^{k}}\quad &\min_{j\in\{1,\ldots,k\}}{p_j}\\
    \rm{s.t.}\quad &\sum_{i=1}^{k}{p_i} \leq P_\Sigma\\
    &p_i \geq 0, \quad  i = 1,\ldots,k
\end{aligned}
\end{equation}
The problem in \eqref{eq:wcevarsol} is a minimax problem over a scaled $k$-simplex. Therefore, its solution is an equalizer rule \cite{poor}, leading to
$p_1^*=p_2^*=\cdots=p_k^*$ with $\sum_{i=1}^{k}{p_i^*}=P_\Sigma$. Hence, the solution of \eqref{eq:wcevarsol} is given by $p_i^*={P_\Sigma}/{k}$ for $i=1,\ldots,k$; that is, the optimal power allocation strategy to maximize (minimize) the lower (upper) \emph{bound} on the minimum (maximum) eigenvalue of the FIM (CRLB) is the equal power allocation strategy. Consequently, the lower bound on the minimum eigenvalue of ${\bf{I}}({\bf{X}};\boldsymbol{\theta})$ becomes $\lambda_{\min}\{{\bf{J}}\}P_\Sigma/k$.



\vspace{0.1cm}

\noindent\textit{\textbf{Remark~1:} The preceding analysis indicates that the equal power allocation strategy solves the problem of maximizing a \textbf{lower bound} on the objective function in \eqref{eq:wcevaropt}. Hence, it does not necessarily yield the optimal power allocation strategy. (The numerical example in Section~\ref{sec:WCEcriSim} illustrates this fact.)}


\subsection{Worst-Case Coordinate Error Variance Criterion}

As an alternative measure of robustness, one can consider the worst-case coordinate error variance, which is bounded by the largest diagonal entry of the CRLB; i.e., $\max_{j \in \{1,\ldots,k\}}{[{\bf{I}}^{-1}({\bf{X};\boldsymbol{\theta}})]_{j,j}}$. {This criterion is referred as G-optimality \cite{kalaba,Emery_1998}, and it has the effect of reducing the worst-case error variance as well.}

From \eqref{eq:ObjFn1},
the $j$th diagonal entry of the CRLB can be expressed as
\begin{equation}
    {[{\bf{I}}^{-1}({\bf{X}};\boldsymbol{\theta})]}_{j,j}={\frac{a_{jj}}{p_j}}
\end{equation}
Therefore, the problem of minimizing the maximum diagonal entry of the CRLB can be formulated as
\begin{equation}\label{eq:optProbD}
\begin{aligned}
        \min_{\{p_i\}_{i=1}^{k}}\:&\max_{j\in\{1,\ldots,k\}}{{\frac{a_{jj}}{p_j}}}\\
        \rm{s.t.}\quad&\sum_{i=1}^{k}{p_i} \leq P_\Sigma\\
        &p_i \geq 0,\quad i=1,\ldots,k
\end{aligned}
\end{equation}
The problem in \eqref{eq:optProbD} is a convex optimization problem, it can be shown that the solution of  \eqref{eq:optProbD} satisfies $\sum_{i=1}^{k}{p_i^*}=P_\Sigma$ and 
\begin{equation}\label{eq:equalizeD}
    \frac{a_{ii}}{p_i^{*}}=\alpha,\quad \forall i \in \{1,\ldots,k\}
\end{equation}
where $\alpha$ is a constant (i.e., an equalizer solution \cite{minimax}).
Then, parameter $\alpha$ in \eqref{eq:equalizeD} obtained from
    \begin{equation}
        \sum_{i=1}^{k}{p_i^*}=\frac{1}{\alpha}\sum_{i=1}^{k}{a_{ii}}=P_\Sigma
    \end{equation}
which yields
\begin{equation}
        \alpha=\frac{{\rm{tr}}\{{\bf{A}}\}}{P_\Sigma}\,\cdot
\end{equation}
Hence, the optimal power allocation strategy is given by
\begin{equation}\label{eq:wccevar_opt}
    p_i^*=\frac{P_\Sigma\,a_{ii}}{{\rm{tr}}\{{\bf{A}}\}},\quad i=1,\ldots,k
\end{equation}
When the optimal power allocation strategy is employed, the diagonal entries of the CRLB are the same, and the worst-case coordinate error variance becomes $\alpha={{\rm{tr}}\{{\bf{A}}\}}/{P_\Sigma}$.

\subsection{Average Fisher Information Criterion}\label{sec:AvgFI}

When the aim is to estimate a vector of parameters, the average Fisher information indicates the overall usefulness of the observation vector to estimate the parameter vector. The informativeness of the observation vector to estimate the $i$th parameter corresponds to the $i$th diagonal entry of the FIM. Therefore, the average Fisher information is related to the trace of the FIM. Accordingly, the optimal power allocation problem for maximizing the trace of the FIM is formulated as follows:
\begin{equation}\label{eq:AvgFIprob}
    \begin{aligned}
        \max_{\{p_i\}_{i=1}^{k}}\quad&{{\rm{tr}}\{{\bf{I}}({\bf{X}};\boldsymbol{\theta})\}}\\
        \rm{s.t.}\quad&\sum_{i=1}^{k}{p_i} \leq P_\Sigma\\
        &p_i \geq 0,\quad i=1,\ldots,k
    \end{aligned}
\end{equation}
From \eqref{eq:fimx} and \eqref{eq:jdef}, the objective function in \eqref{eq:AvgFIprob} can be rewritten in terms of the known matrices as
\begin{equation}\label{eq:ObjAvgFI}
    {{\rm{tr}}\{{\bf{I}}({\bf{X}};\boldsymbol{\theta})\}}
    ={{\rm{tr}}\{{\bf{P}}{\bf{J}}{\bf{P}}\}}
    ={\sum_{i=1}^{k}{p_i\,j_{ii}}}
\end{equation}
where $j_{ii}\triangleq[{\bf{J}}]_{i,i}$. Based on \eqref{eq:ObjAvgFI}, the problem in \eqref{eq:AvgFIprob} can be converted to a linear program (LP) by defining
\begin{align}\label{eq:Defp}
{\bf{p}}&\triangleq {\bf{diag}}({\bf{P}}{\bf{P}}^T)=[p_1\ldots p_k]^T\\\label{eq:Defj}
{\bf{j}}&\triangleq{\bf{diag}}({\bf{J}})=[j_{11}\ldots j_{kk}]^T,
\end{align}
and expressing \eqref{eq:AvgFIprob} as
\begin{equation}\label{optavgFI}
    \begin{aligned}
    \max_{\bf{p}}\quad&{{\bf{j}}^T{\bf{p}}}\\
    \rm{s.t.}\quad&{\bf{1}}^T{\bf{p}}\leq P_\Sigma\\
    &{\bf{p}}\geq 0
    \end{aligned}
\end{equation}
The solution of \eqref{optavgFI} is provided in the following proposition.

\begin{claim}
Let $i^*$ denote the index of the maximum element of ${\bf{j}}$ in \eqref{eq:Defj}; i.e., $i^*=\arg\max_{l \in \{1,\ldots,k\}} {j_{ll}}$.
Then, the optimal power allocation strategy that maximizes the average Fisher information under the sum-power constraint is given by ${\bf{p}}^*= [p_1^*\cdots p_k^*]^T$, where
\begin{equation}\label{eq:AvgFIoptSol}
    p_{i}^{*} = \begin{cases}
        P_\Sigma, &i = i^* \\
        0, &\rm{otherwise}
    \end{cases}
\end{equation}
for $i=1,\ldots,k$. (In case of multiple maxima, the indices of any non-empty subset can be selected.)
\end{claim}
Proposition~1 states that in order to maximize the average FI, the whole power must be allocated to the parameter corresponding to the maximum diagonal entry of $\bf{J}$. When the optimal power adaptation strategy is used, the average Fisher information achieves a maximum value of ${P_\Sigma}\max_i{j_{ii}}$.



\subsection{Worst-Case Coordinate Fisher Information Criterion}\label{sec:WCCFI}

Depending on the system properties, the observation vector can be more informative about some parameters and less informative about the others. Consequently, the estimation performance of individual parameters can vary to some extent. Such variations can be undesirable as certain performance requirements should be satisfied for estimation of all parameters. To alleviate this effect, one approach is to maximize the minimum Fisher information contained the observation vector w.r.t. individual parameters. Such an optimization increases the robustness of estimation against accuracy variations.

The minimum Fisher information contained in the observation vector w.r.t. individual parameters is called the worst-case coordinate FI, which corresponds to the minimum diagonal entry of the FIM, that is, $\min_{i \in \{1,...,k\}}{{[{\bf{I}}({\bf{X}};\boldsymbol{\theta})]}_{i,i}}$.
Based on this objective function, the following maximization problem is defined under the sum-power constraint:
\begin{equation}\label{eq:wcfiopt}
    \begin{aligned}
        \max_{\{p_i\}_{i=1}^{k}}\quad&\min_{i \in \{1,...,k\}}{{[{\bf{I}}({\bf{X}};\boldsymbol{\theta})]}_{i,i}}\\
        \rm{s.t.}\quad&\sum_{i=1}^{k}{p_i} \leq P_\Sigma\\
        &p_i\geq 0,\quad i=1,\ldots,k
    \end{aligned}
\end{equation}
Based on \eqref{eq:ObjAvgFI}, the objective function in \eqref{eq:wcfiopt} can be written as
\begin{equation}
    \min_i{{[{\bf{I}}({\bf{X}};\boldsymbol{\theta})]}_{i,i}}=\min_i\,{p_i\,j_{ii}}
\end{equation}
Then, \eqref{eq:wcfiopt} is observed to have a very similar form to the problem in \eqref{eq:wcevarsol}. Hence, the same steps can be followed and it can be shown that the solution of \eqref{eq:wcfiopt} satisfies $\sum_{i=1}^{k}{p_i^*}=P_\Sigma$ and
\begin{equation}
    p_{i}^{*}\,j_{ii}=\tilde{\alpha},\quad\forall i \in \{1,\ldots,k\}
\end{equation}
where $\tilde{\alpha}$ is a constant that is specified by
    \begin{equation}\label{eq:totPowF}
        \sum_{i=1}^{k}{p_i^*}=\sum_{i=1}^{k}\frac{\tilde{\alpha}}{j_{ii}}=P_\Sigma
    \end{equation}
From \eqref{eq:totPowF}, $\tilde{\alpha}$ is obtained as
\begin{equation}
\tilde{\alpha}=\frac{P_\Sigma}{\sum_{i=1}^{k}{\frac{1}{j_{ii}}}}\,.
\end{equation}
Therefore, the optimal power allocation strategy becomes
\begin{equation}\label{eq:FinalCriSol}
    p_{i}^{*}=\frac{P_\Sigma}{j_{ii}\sum_{l=1}^{k}{\frac{1}{j_{ll}}}}\,,\quad i=1,\ldots,k
\end{equation}
When the optimal power allocation strategy is used, the worst case Fisher information achieves a maximum value of ${P_\Sigma}\big/{\sum_{l=1}^k{\frac{1}{j_{ll}}}}$.
It is noted that the optimal power allocation strategy in \eqref{eq:FinalCriSol} equalizes the Fisher information contained in $\bf{X}$ w.r.t. the $i$th element of $\boldsymbol{\theta}$ for all $i\in\{1,\ldots,k\}$; that is, $p_i^*j_{ii}={P_\Sigma}\big/{\sum_{l=1}^k{\frac{1}{j_{ll}}}}$ for all $i\in\{1,\ldots,k\}$.

\section{Extensions}\label{sec:Extend}

\subsection{Presence of Nuisance Parameters}\label{sec:Nuis}

In some vector parameter estimation scenarios, only a subset of parameters can be of interest for estimation purposes. 
Let only $r$ out of the $k$ parameters be relevant and the remaining $k-r$ parameters be nuisance parameters. We assume that the nuisance parameters must be transmitted with unit power and power adaptation is not available for them. Without loss of generality, we can arrange the vector of parameters to be transmitted as
\begin{equation}
    {\boldsymbol{\theta}}=\begin{bmatrix}
    {\boldsymbol{\theta}}_\gamma\\
    {\boldsymbol{\theta}}_\sigma
    \end{bmatrix}
\end{equation}
where ${\boldsymbol{\theta}}_\gamma\in{\mathbb{R}}^{r}$ denotes the vector of relevant parameters, and ${\boldsymbol{\theta}}_\sigma\in{\mathbb{R}}^{k-r}$ represents the vector of nuisance parameters. Then, the power allocation matrix becomes
\begin{equation}\label{eq:NewP}
    {\bf{P}}=\begin{bmatrix}
    {\bf{P}}_\gamma&{\bf{0}}\\
    {\bf{0}}&{\bf{I}}_{k-r}
    \end{bmatrix}
\end{equation}
where ${\bf{I}}_{k-r}$ denotes the $(k-r)\times(k-r)$ identity matrix. Under the same system model, matrix $\bf{J}$ defined in \eqref{eq:jdef} can be expressed as
\begin{equation}\label{eq:NewJ}
    {\bf{J}}=\begin{bmatrix}
        {\bf{J}}_\gamma&{\bf{B}}\\
        {\bf{B}}^T&{\bf{J}}_\sigma
    \end{bmatrix}
\end{equation}
where ${\bf{J}}_\gamma\in{\mathbb{R}}^{r\times r}$ and ${\bf{J}}_\sigma\in{\mathbb{R}}^{(k-r)\times (k-r)}$ are the components of $\bf{J}$ corresponding to the parameters of interest and the nuisance parameters, respectively, and matrix ${\bf{B}}\in\mathbb{R}^{r\times (k-r)}$ and its transpose are the cross-terms. Similarly, matrix $\bf{A}$ in \eqref{eq:jdef}, i.e., $\bf{A}\triangleq\bf{J}^{-1}$, can also be arranged as
\begin{equation}\label{eq:NewA}
    {\bf{A}}=\begin{bmatrix}
    {\bf{A}}_\gamma&{\bf{C}}\\
    {\bf{C}}^T&{\bf{A}}_\sigma
    \end{bmatrix}
\end{equation}
where ${\bf{C}}\in\mathbb{R}^{r\times (k-r)}$ and its transpose are the cross-terms. Based on \eqref{eq:NewP}, \eqref{eq:NewJ}, and \eqref{eq:NewA}, the FIM and the CRLB can be expressed as
\begin{align}
    {\bf{I}}({\bf{X}};{\boldsymbol{\theta}})&=
    {\bf{P}}{\bf{J}}{\bf{P}}=
    \begin{bmatrix}
    \bf{P}_\gamma\bf{J}_\gamma\bf{P}_\gamma&\bf{P}_\gamma\bf{B}\\
    {\bf{B}}^T\bf{P}_\gamma&\bf{J}_\sigma
    \end{bmatrix}\\
    {\bf{I}}^{-1}({\bf{X}};{\boldsymbol{\theta}})&=
    {\bf{P}}^{-1}{\bf{A}}{\bf{P}}^{-1}=
    \begin{bmatrix}
    \bf{P}^{-1}_\gamma\bf{A}_\gamma\bf{P}^{-1}_\gamma&\bf{P}^{-1}_\gamma\bf{C}\\
    {\bf{C}}^T\bf{P}^{-1}_\gamma&\bf{A}_\sigma
    \end{bmatrix}
\end{align}
The related terms of the FIM and the CRLB are the ones involving only the parameters of interest. In this setting, only the first $r$ rows and the first $r$ columns are taken into account; that is,
\begin{align}
    {\bf{I}}_\gamma({\bf{X}};{\boldsymbol{\theta}})&={\bf{P}}_\gamma\bf{J}_\gamma{\bf{P}}_\gamma\\
    {\bf{I}}^{-1}_\gamma({\bf{X}};{\boldsymbol{\theta}})&={\bf{P}}^{-1}_\gamma\bf{A}_\gamma{\bf{P}}^{-1}_\gamma
\end{align}
where ${\bf{A}}_\gamma=({\bf{J}}_\gamma-{\bf{B}}{\bf{J}}_\sigma^{-1}{\bf{B}}^{T})^{-1}$ \cite{shenwin}. As seen from above, the power allocation strategies developed in Section~\ref{sec:approaches} (which are developed in the absence of nuisance parameters) can also be used in this case.


\subsection{Extension to Nonlinear Model}\label{sec:NonLin}

In some practical applications, the linear system model in \eqref{eq:sysmodel} may not be valid, and the parameter vector, after power adaptation, can be processed by a nonlinear transformation ${\bf{f}}(\cdot)$ as follows:
\begin{equation}\label{eq:nonLin}
    {\bf{X}}={\bf{f}}({\bf{P}}\boldsymbol{\theta})+{\bf{N}}
\end{equation}
In this case, the FIM w.r.t. parameter $\boldsymbol{\theta}$ can be expressed in the same form as \eqref{eq:fimx} after replacing $\bf{F}$ with the Jacobian of the vector valued function ${\bf{f}}(\cdot)$ \cite[Lemma~4]{zamir}. More explicitly, let ${\boldsymbol{\phi}}\triangleq\bf{P}\boldsymbol{\theta}$ in \eqref{eq:nonLin}, and the Jacobian of ${\bf{f}}(\boldsymbol{\phi})$ w.r.t. its argument $\boldsymbol{\phi}$ is given as
\begin{equation}\label{eq:newF}
    \bf{F} \triangleq \begin{bmatrix}
    \frac{\partial f_1}{\partial\phi_1}&\frac{\partial f_2}{\partial\phi_1}&\hdots&\frac{\partial f_n}{\partial\phi_1}\\
    \frac{\partial f_1}{\partial\phi_2}&\frac{\partial f_2}{\partial\phi_2}&\hdots&\frac{\partial f_n}{\partial\phi_2}\\
    \vdots&\vdots&\ddots&\vdots\\
    \frac{\partial f_1}{\partial\phi_k}&\frac{\partial f_2}{\partial\phi_k}&\hdots&\frac{\partial f_n}{\partial\phi_k}
    \end{bmatrix}
\end{equation}
If ${\bf{f}}(\cdot)$ is continuously differentiable w.r.t. $\boldsymbol{\phi}$ and $\bf{F}$ in \eqref{eq:newF} does not depend on $p_i$'s for $i=1,\dots,k$, $\bf{F}$ in \eqref{eq:newF} can be substituted into \eqref{eq:fimx}, and the developed techniques can be employed without further modification for power adaptation in the presence of a nonlinear system model, as well. If $\bf{F}$ depends on $p_i$'s, \eqref{eq:newF} is still valid; however, the objective functions should be modified accordingly, leading to possibly nonconvex optimization problems. In that case, numerical methods can be employed. On the other hand, if ${\bf{f}}(\cdot)$ is not continuously differentiable w.r.t. $\boldsymbol{\phi}$, further analysis is required and new techniques should be developed.


\section{Numerical Results}\label{sec:results}

In this section, we provide numerical examples for the optimal power allocation strategies in Section~\ref{sec:approaches}. 
In all cases, the equal power allocation strategy is also implemented for comparison purposes. The noise is modeled as a zero-mean Gaussian random vector with independent components; that is, ${\bf{N}}\sim \mathcal{N}(\boldsymbol{0},\Sigma)$, where $\Sigma={\rm{Diag}}(\sigma_1^2,\hdots,\sigma_n^2)$. For this noise model, the FIM of $\bf{N}$ in \eqref{eq:In} is obtained as
\begin{equation}
    {\bf{I}}({\bf{N}})=\Sigma^{-1}={\rm{Diag}}\left(1/\sigma_1^2,\hdots,1/\sigma_n^2\right)
\end{equation}
In the simulations, 
$\sigma_i^2$'s are set to $\sigma_i^2=10^{-7+3(i-1)/(n-1)}$ for $i=1,\ldots,n$.
The dimension of the parameter vector, $k$, is varied between $2$ and $30$, and the dimension of the observation vector, $n$, is taken to be equal to the number of parameters, i.e., $k=n$. Also, for matrix $\bf{F}$ in \eqref{eq:sysmodel}, two different scenarios are considered. In the first scenario, ${\bf{F}}={\bf{F}}_1$, where ${\bf{F}}_1$ is  the $k\times k$ identity matrix ($k=n$), that is, ${\bf{F}}_1 = {\bf{I}}_{k\times k}$. In this scenario, we can observe the effects of power adaptation on the estimation performance when the main source of error is additive noise.
In the second scenario, ${\bf{F}}={\bf{F}}_2$, which is specified as
\begin{equation}
        {\bf{F}}_2 = {\bf{I}}_{k\times k} + \kappa{\bf{V}}^T
\end{equation}
with
\begin{equation}
    \begin{aligned}
        \kappa &= \frac{\lVert{\bf{I}}_{k\times k}\rVert_{\rm{F}}}{\lVert{\bf{V}}\rVert_{\rm{F}}}\\
        {\bf{V}} &= \begin{bmatrix}
        1 & 1 & 1 & \hdots & 1\\
        1 & 1+\epsilon & (1+\epsilon)^2 & \hdots & (1+\epsilon)^{k-1}\\
        1 & 1+2\epsilon & (1+2\epsilon)^2 & \hdots & (1+2\epsilon)^{k-1}\\
        \vdots &  \vdots & \vdots &  & \vdots\\
        1 & 1.5 & 1.5^2 & \hdots & 1.5^{k-1}
        \end{bmatrix}\\
        \epsilon &= \frac{0.5}{k-1}
    \end{aligned}
\end{equation}
That is, ${\bf{F}}_2$ is the sum of the $k\times k$ identity matrix and the transpose of a normalized $k\times k$ Vandermonde matrix, where
the normalization factor $\kappa$ makes sure that the Frobenius norms of the added matrices are equal.
In this scenario, the entries of the system matrix $\bf{F}$ differ from each other significantly, which implies that the system affects the individual parameters  differently. 



In the following, the optimal power allocation strategies are obtained according to the Fisher information based criteria in Section~\ref{sec:approaches} for the considered simulation setup, and the performance metrics are plotted against the dimension of the parameter vector, $k$, under a unit sum-power constraint, that is, $P_{\Sigma}=0$ dB.

\subsection{Results for Average MSE Criterion}

In this case, the problem in \eqref{eq:avgmseopt} is considered, and the CRLBs achieved by the optimal power allocation strategy in \eqref{eq:avgmse} and by the equal power allocation strategy (i.e., $p_i^*=P_{\Sigma}/k$, $i=1,\ldots,k$) are plotted versus $k$ in Fig.~\ref{fig:avgmse}. It is noted that as the dimension of the parameter vector increases, the CRLB on the average MSE increases for both optimal and equal power allocation strategies except for the slight initial decrease in the optimal strategy for ${\bf{F}}={\bf{F}}_2$. It is also observed that the optimal power allocation strategy consistently outperforms the equal power allocation strategy for both system matrices. As an example, for $\bf{F} = {\bf{F}}_2$, the CRLB is around $10^{-4}$ when $k = n = 7$ for the equal power allocation strategy, and the same level of CRLB is attained when $k = n = 14$ for the optimal power allocation strategy. Hence, significant improvements can be achieved by the optimal power allocation strategy.

\begin{figure}
    \centering
    \includegraphics[width=0.75\textwidth]{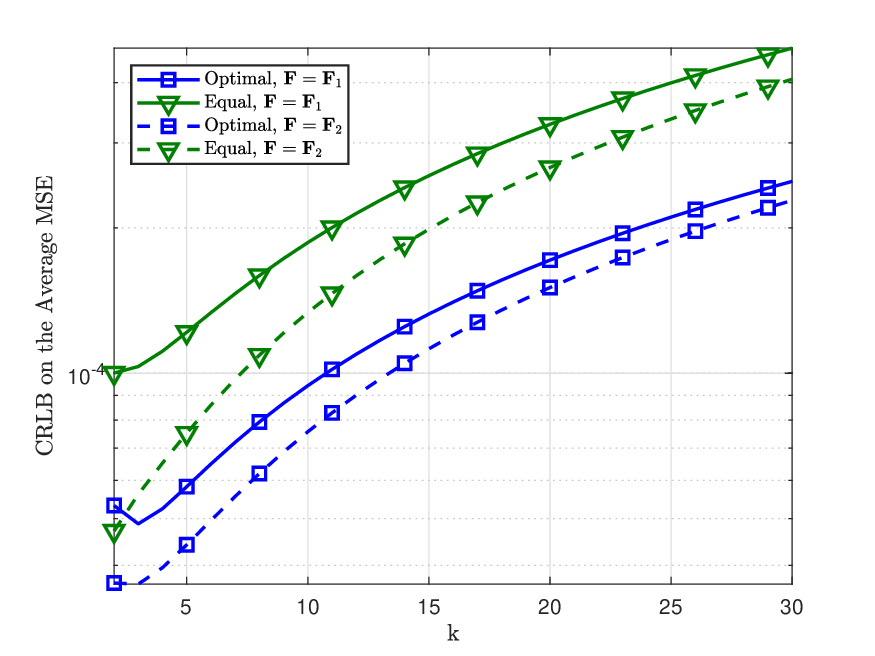}
    \caption{CRLB on the average MSE versus $k$ for the equal and optimal power allocation strategies.}
    \label{fig:avgmse}
\end{figure}


\subsection{Results for Shannon Information Criterion}

For this criterion, the problem in \eqref{eq:shannonInfoOpt} (equivalently, \eqref{shannoninfosol}) is considered, which leads to the solution in \eqref{shannonOptimal}. That is, the optimal and equal power allocation strategies yield the same solution in this case. The Shannon information achieved by the optimal (equal) power allocation strategy is plotted versus $k$ in Fig.~\ref{fig:shannoninfo}. It is observed that the Shannon information increases as the dimension of the parameter vector, $k$, increases. The increase in Shannon information is linear for both system matrices, and the achieved Shannon information scores are nearly the same for ${\bf{F}} = {\bf{F}}_1$ and ${\bf{F}} = {\bf{F}}_2$.




\begin{figure}
    \centering
    \includegraphics[width=0.75\textwidth]{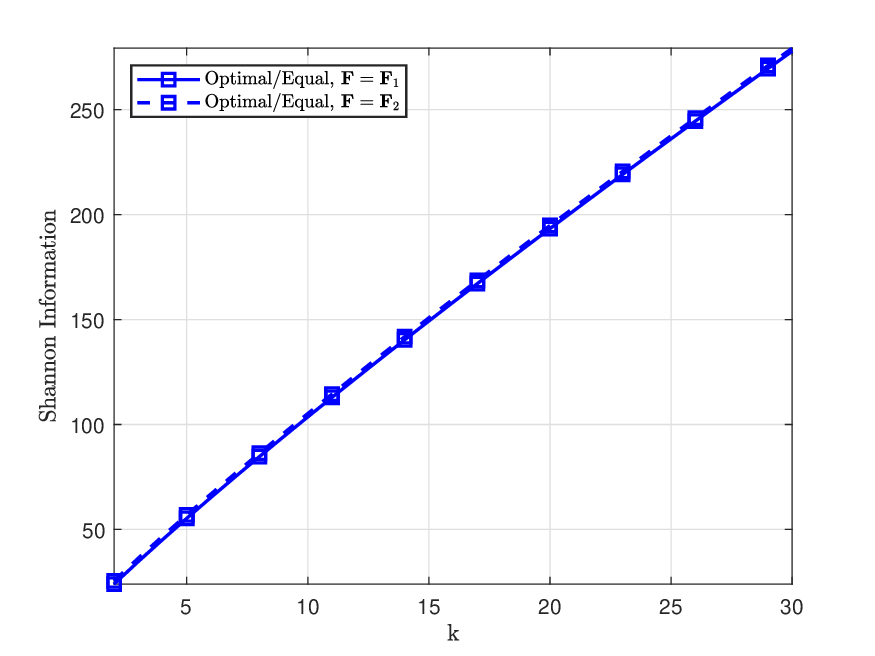}
    \caption{Shannon information versus $k$ for the optimal (equal) power allocation strategy.}
    \label{fig:shannoninfo}
\end{figure}

\subsection{Results for Worst-Case Error Variance Criterion}\label{sec:WCEcriSim}

In this case, the problem in \eqref{eq:wcevaropt} and the alternative problem in \eqref{eq:wcevarsol} are solved. The solution of \eqref{eq:wcevaropt} is obtained via the multistart global optimization algorithm in MATLAB. On the other hand, the equal power allocation strategy is the solution of \eqref{eq:wcevarsol}, as shown in Section~\ref{sec:WCEVC}. In Fig.~\ref{fig:worstvar}, the maximum eigenvalues of the CRLBs achieved by the optimal and equal power allocation strategies are plotted versus $k$. It can be seen in Fig.~\ref{fig:worstvar} that the optimal power allocation strategy can significantly outperform the equal power allocation strategy, and the difference between the two power allocation strategies increases as the number of parameters increases. One implication of this result is that power adaptation can get more effective when there exist more parameters to estimate. In addition, it is noted that maximizing the lower bound on the eigenvalues of the CRLB is not sufficient to obtain the optimal power allocation strategy, as stated in Remark~1.

\begin{figure}
    \centering
    \includegraphics[width=0.75\textwidth]{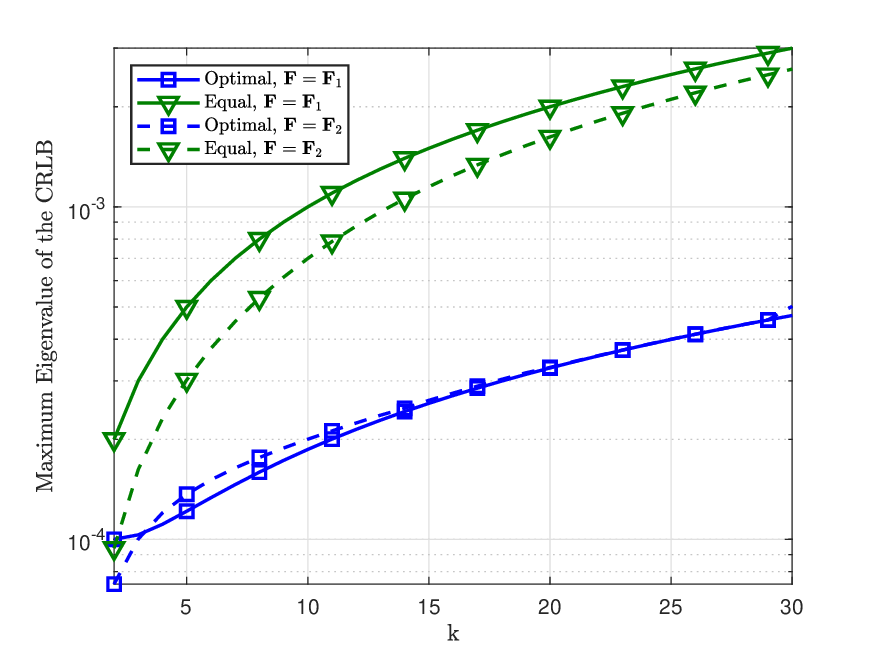}
    \caption{Maximum eigenvalue of the CRLB (inverse FIM) versus $k$ for the optimal and equal power allocation strategies.}
    \label{fig:worstvar}
\end{figure}

\subsection{Results for Worst-Case Coordinate Error Variance Criterion}

For this criterion, the problem in \eqref{eq:optProbD} is considered, which leads to the optimal solution in \eqref{eq:wccevar_opt}. The largest diagonal entry of the CRLB is plotted versus $k$ for both the optimal solution and the equal power allocation strategy in Fig.~\ref{fig:worstcoorvar}. It is noted that the trend is similar to that in Fig.~\ref{fig:worstvar}. Namely, the benefits of optimal power adaptation are observed for the worst-case coordinate error variance criterion, as well. 

\begin{figure}
    \centering
    \includegraphics[width=0.75\textwidth]{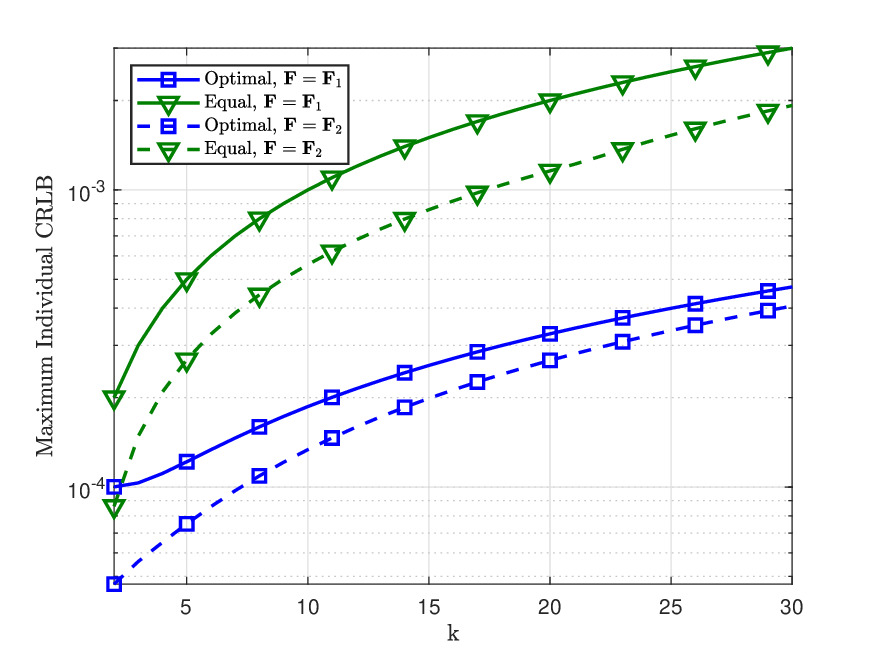}
    \caption{The largest diagonal entry of the CRLB (worst-case coordinate CRLB) versus $k$ for the optimal and equal power allocation strategies.}
    \label{fig:worstcoorvar}
\end{figure}

\subsection{Results for Average Fisher Information Criterion}

In this case, we focus on the problem in \eqref{eq:AvgFIprob}, the solution of which is provided by \eqref{eq:AvgFIoptSol} in Proposition~1. The impact of the dimension of the parameter vector, $k$, on the average Fisher information is shown in Fig.~\ref{fig:averageFI} for both the optimal solution in \eqref{eq:AvgFIoptSol} and the equal power allocation strategy. It is observed that the average Fisher information rapidly decreases with $k$ when $k\leq 10$ for both the optimal and equal power allocation strategies.
While the optimal power allocation strategy is superior to the equal power allocation for all values of $k$, significant enhancements are observed for large values of $k$.

\begin{figure}
    \centering
    \includegraphics[width=0.75\textwidth]{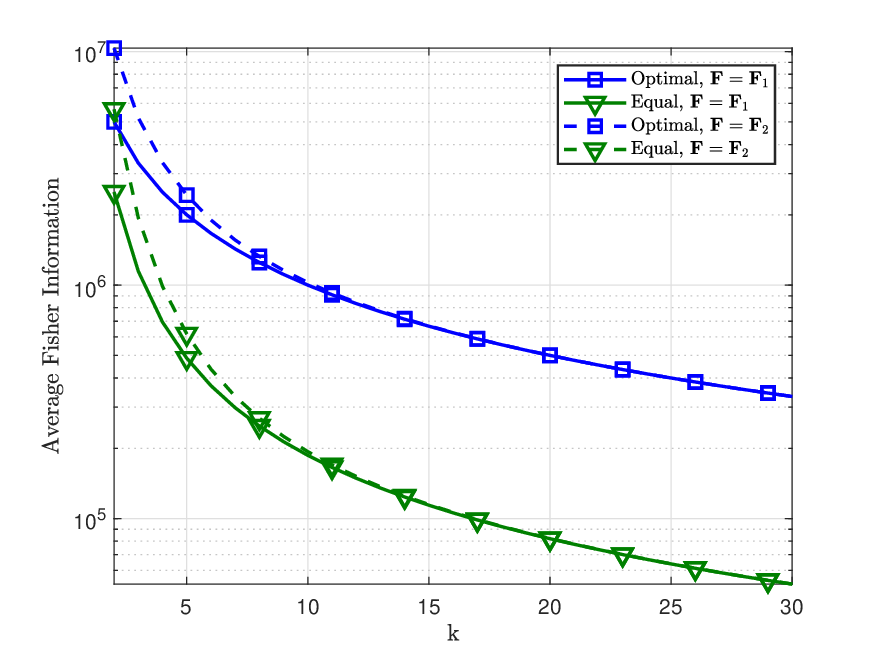}
    \caption{Average Fisher information versus $k$ for the optimal and equal power allocation strategies.}
    \label{fig:averageFI}
\end{figure}

\subsection{Results for Worst Case Coordinate Fisher Information}

In this scenario, the minimum diagonal entry of the FIM is maximized as in \eqref{eq:wcfiopt}, leading to the optimal power allocation strategy in \eqref{eq:FinalCriSol}. The minimum diagonal entry of the FIM is plotted versus $k$ for both the optimal and equal power allocation strategies in Fig.~\ref{fig:worstCoordFI}. When ${\bf{F}}={\bf{F}}_2$, the worst-case coordinate Fisher information rapidly decreases for small $k$, while the trend is more steady when ${\bf{F}}={\bf{F}}_1$.
The decrease in worst-case Fisher information slows down for large values of $k$.
Overall, the impact of power adaptation can be observed more clearly when $k$ is large. 

\begin{figure}
    \centering
    \includegraphics[width=0.75\textwidth]{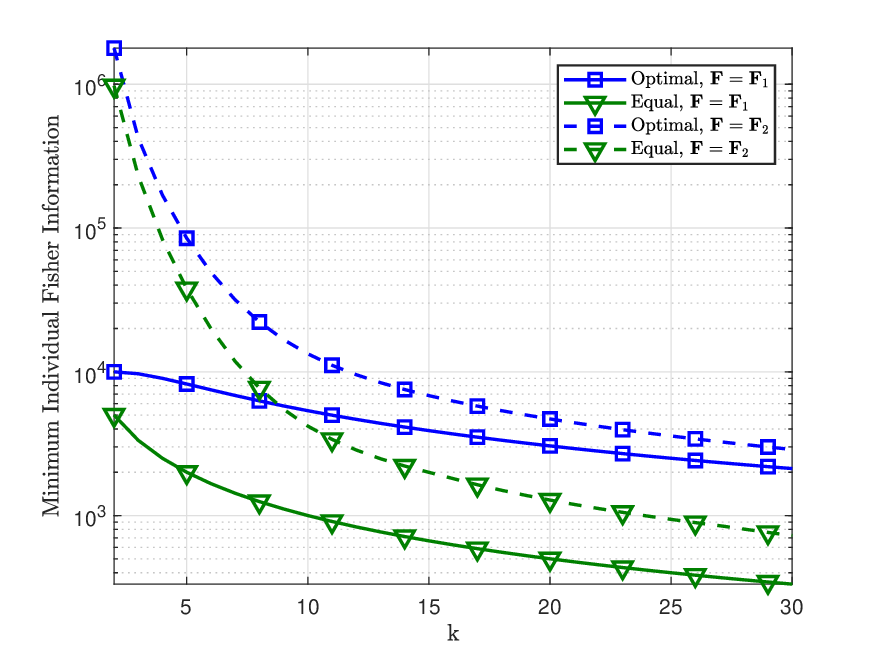}
    \caption{The minimum diagonal entry of the FIM (worst-case coordinate Fisher information) versus $k$ for the optimal and equal power allocation strategies.}
    \label{fig:worstCoordFI}
\end{figure}

${}$

It is noted from the simulation results that when the dimensions of the parameter and observation vectors are large, power adaptation becomes more critical and the optimal power allocation strategies can provide more significant improvements over the equal power allocation strategy. In addition, the trends show that power adaptation can mitigate the adverse effects of increases in the dimension of the parameter vector when the observation vector has the same dimension as the parameter vector.

\section{Conclusion}\label{sec:conclusion}

The optimal power allocation problem has been investigated for vector parameter estimation in the absence of prior information according to various Fisher information based optimality criteria. After deriving the FIM for a linear observation model, six different optimal power allocation problems have been formulated. Then, some closed-form solutions have been provided based on optimization theoretic approaches. It has been shown that the proposed power allocation strategies are also valid for nonlinear system models under certain conditions and in the presence of nuisance parameters. Numerical examples have shown that the use of the optimal power allocation strategies can provide significant improvements in estimation performance over the equal power allocation strategy.

\ifCLASSOPTIONcaptionsoff
  \newpage
\fi



\bibliographystyle{IEEEtran}
\bibliography{references}
%



%





\end{document}